\documentclass[12pt]{amsart}
\usepackage{amsmath,amssymb,amsbsy,amsfonts,latexsym,amsopn,amstext,cite,
                                               amsxtra,euscript,amscd,bm}
\usepackage{url}
\usepackage[colorlinks,linkcolor=blue,anchorcolor=blue,citecolor=blue,backref=page]{hyperref}
\usepackage{color}
\usepackage{graphics,epsfig}
\usepackage{graphicx}
\usepackage{float}

\hypersetup{breaklinks=true}

\renewcommand*{\backref}[1]{}
\renewcommand*{\backrefalt}[4]{%
    \ifcase #1 (Not cited.)%
    \or        (p.\,#2)%
    \else      (pp.\,#2)%
    \fi}

\usepackage{mathtools}
\usepackage{todonotes}

\begin{document}

\newtheorem{theorem}{Theorem}
\newtheorem{lemma}[theorem]{Lemma}
\newtheorem{claim}[theorem]{Claim}
\newtheorem{cor}[theorem]{Corollary}
\newtheorem{prop}[theorem]{Proposition}
\newtheorem{definition}[theorem]{Definition}
\newtheorem{question}[theorem]{Question}
\newcommand{\hh}{{{\mathrm h}}}

\numberwithin{equation}{section}
\numberwithin{theorem}{section}
\numberwithin{table}{section}

\def\sssum{\mathop{\sum\!\sum\!\sum}}
\def\ssum{\mathop{\sum\ldots \sum}}
\def\iint{\mathop{\int\ldots \int}}

\def\squareforqed{\hbox{\rlap{$\sqcap$}$\sqcup$}}
\def\qed{\ifmmode\squareforqed\else{\unskip\nobreak\hfil
\penalty50\hskip1em\null\nobreak\hfil\squareforqed
\parfillskip=0pt\finalhyphendemerits=0\endgraf}\fi}

\newfont{\teneufm}{eufm10}
\newfont{\seveneufm}{eufm7}
\newfont{\fiveeufm}{eufm5}
%
%
\newfam\eufmfam
     \textfont\eufmfam=\teneufm
\scriptfont\eufmfam=\seveneufm
     \scriptscriptfont\eufmfam=\fiveeufm
%
%
\def\frak#1{{\fam\eufmfam\relax#1}}

\newcommand{\bflambda}{{\boldsymbol{\lambda}}}
\newcommand{\bfmu}{{\boldsymbol{\mu}}}
\newcommand{\bfxi}{{\boldsymbol{\xi}}}
\newcommand{\bfrho}{{\boldsymbol{\rho}}}

\def\fK{\mathfrak K}
\def\fT{\mathfrak{T}}

\def\fA{{\mathfrak A}}
\def\fB{{\mathfrak B}}
\def\fC{{\mathfrak C}}
\def\fM{{\mathfrak M}}

\def \balpha{\bm{\alpha}}
\def \bbeta{\bm{\beta}}
\def \bgamma{\bm{\gamma}}
\def \blambda{\bm{\lambda}}
\def \bchi{\bm{\chi}}
\def \bphi{\bm{\varphi}}
\def \bpsi{\bm{\psi}}
\def \bomega{\bm{\omega}}
\def \btheta{\bm{\vartheta}}

\def\nrp#1{\left\|#1\right\|_p}
\def\nrq#1{\left\|#1\right\|_m}
\def\nrqk#1{\left\|#1\right\|_{m_k}}
\def\Ln#1{\mbox{\rm {Ln}}\,#1}
\def\nd{\hspace{-1.2mm}}
\def\ord{{\mathrm{ord}}}
\def\Cc{{\mathrm C}}
\def\Pb{\,{\mathbf P}}
\def\Oes{\cO_{e,s}}
\def\Oesf{\cO_{e,s,f}}
\def\Oesfx{\cO_{e,s,f}(x)}
\def\Oesx{\cO_{e,s}(x)}
\def\OAs{\cO_{\cA,s}}
\def\Oet{\cO_{e,t}}
\def\Oetg{\cO_{e,t,g}}
\def\Oetx{\cO_{e,t}(x)}
\def\Oef{\fO_{e,f}}
\def\Oeg{\fO_{e,g}}
\def\fO{\mathfrak{O}}
\def\fQ{\mathfrak{Q}}

\def \bxi{\bm{\xi}}

\def\eqref#1{(\ref{#1})}

\def\vec#1{\mathbf{#1}}


\def\cA{{\mathcal A}}
\def\cB{{\mathcal B}}
\def\cC{{\mathcal C}}
\def\cD{{\mathcal D}}
\def\cE{{\mathcal E}}
\def\cF{{\mathcal F}}
\def\cG{{\mathcal G}}
\def\cH{{\mathcal H}}
\def\cI{{\mathcal I}}
\def\cJ{{\mathcal J}}
\def\cK{{\mathcal K}}
\def\cL{{\mathcal L}}
\def\cM{{\mathcal M}}
\def\cN{{\mathcal N}}
\def\cO{{\mathcal O}}
\def\cP{{\mathcal P}}
\def\cQ{{\mathcal Q}}
\def\cR{{\mathcal R}}
\def\cS{{\mathcal S}}
\def\cT{{\mathcal T}}
\def\cU{{\mathcal U}}
\def\cV{{\mathcal V}}
\def\cW{{\mathcal W}}
\def\cX{{\mathcal X}}
\def\cY{{\mathcal Y}}
\def\cZ{{\mathcal Z}}
\newcommand{\rmod}[1]{\: \mbox{mod} \: #1}

\def\cg{{\mathcal g}}

\def\vr{\mathbf r}

\def\e{{\mathbf{\,e}}}
\def\ep{{\mathbf{\,e}}_p}
\def\eq{{\mathbf{\,e}}_q}

\def\Tr{{\mathrm{Tr}}}
\def\Nm{{\mathrm{Nm}}}

 \def\SS{{\mathbf{S}}}

\def\lcm{{\mathrm{lcm}}}
\def\Res{{\mathrm{Res}}}

\def\({\left(}
\def\){\right)}
\def\l|{\left|}
\def\r|{\right|}
\def\fl#1{\left\lfloor#1\right\rfloor}
\def\rf#1{\left\lceil#1\right\rceil}

\def\mand{\qquad \mbox{and} \qquad}

\newcommand{\commK}[1]{\marginpar{%
\begin{color}{red}
\vskip-\baselineskip 
\raggedright\footnotesize
\itshape\hrule \smallskip B: #1\par\smallskip\hrule\end{color}}}

\newcommand{\commI}[1]{\marginpar{%
\begin{color}{magenta}
\vskip-\baselineskip 
\raggedright\footnotesize
\itshape\hrule \smallskip I: #1\par\smallskip\hrule\end{color}}}

\newcommand{\commT}[1]{\marginpar{%
\begin{color}{blue}
\vskip-\baselineskip 
\raggedright\footnotesize
\itshape\hrule \smallskip I: #1\par\smallskip\hrule\end{color}}}




\hyphenation{re-pub-lished}

\mathsurround=1pt

\def\bfdefault{b}

\def \F{{\mathbb F}}
\def \K{{\mathbb K}}
\def \Z{{\mathbb Z}}
\def \Q{{\mathbb Q}}
\def \R{{\mathbb R}}
\def \C{{\mathbb C}}

\def \Fbar{\overline{\mathbb F}}

\def\Kmnp{\cK_p(m,n)}
\def\Kmnq{\cK_q(m,n)}
\def\Klmnp{\cK_p(\ell, m,n)}
\def\Klmnq{\cK_q(\ell, m,n)}

\def \SALMNq {\cS_q(\balpha;\cL,\cM,\cN)}
\def \SALMNp {\cS_p(\balpha;\cL,\cM,\cN)}

\def\SAMJp{\cS_p(\balpha;\cM,\cJ)}
\def\SAMJq{\cS_q(\balpha;\cM,\cJ)}
\def\SAJq{\cS_q(\balpha;\cJ)}
\def\SAIJp{\cS_p(\balpha;\cI,\cJ)}
\def\SAIJq{\cS_q(\balpha;\cI,\cJ)}

\def\RIJp{\cR_p(\cI,\cJ)}
\def\RIJq{\cR_q(\cI,\cJ)}

\def\TWXJp{\cT_p(\bomega;\cX,\cJ)}
\def\TWXJq{\cT_q(\bomega;\cX,\cJ)}
\def\TWJq{\cT_q(\bomega;\cJ)}

 \def \xbar{\overline x}
  \def \ybar{\overline y}

\newcommand{\ind}{\mathrm{ind\,}}

\title[Identity Testing and Interpolation
from High Powers]{Identity Testing and Interpolation
from High Powers of Polynomials of Large Degree over Finite Fields}

\author[M. Karpinski]{Marek Karpinski}
\address{Department of Computer Science, Bonn University, 
 53113 Bonn, Germany} 
\email{marek@cs.uni-bonn.de}

\author[L. M\'erai]{L\'aszl\'o M\'erai}
\address{Johann Radon Institute for Computational and Applied Mathematics, Austrian Academy of Sciences and Institute of Financial Mathematics and Applied Number Theory,
Johannes Kepler University,  Altenberger Stra\ss e 69, A-4040 Linz, Austria} 
\email{laszlo.merai@oeaw.ac.at}

\author[I. E.~ShparlinskI]{Igor E.~Shparlinski}
\address{Department of Pure Mathematics, University of 
New South Wales, Sydney, NSW 2052 Australia}
\email{igor.shparlinski@unsw.edu.au}

\begin{abstract}
We consider the problem of identity testing and recovering (that is, interpolating)  of a ``hidden'' monic polynomials $f$, given an oracle access to $f(x)^e$ for $x\in\F_q$, where $\F_q$ is the finite field of  $q$ elements and an extension fields access is not permitted.

 The naive 
interpolation algorithm needs $de+1$ queries, where $d =\max\{\deg f, \deg g\}$ 
and thus requires $ de<q$. For a prime $q = p$, we design an algorithm that is asymptotically better in certain cases, especially when $d$ is large. The algorithm is based on 
a result of independent interest in spirit of additive combinatorics. 
It  gives an upper bound on the 
number of values of a rational function of large degree, evaluated on
a short  sequence of consecutive integers, that belong to a small subgroup of $\F_p^*$.  
\end{abstract}



\maketitle
\section{Introduction}

\subsection{Background and previous results} 
The following variant of polynomial identity problem has been motivated
 by some cryptographic applications, see~\cite{BGKS, IKSSS}
 for further discussion  and the references. 
 
 Let $\F_q$ be  the finite field of  $q$ elements of characteristic $p$. 
We  consider the
{\it Identity Testing from Powers\/}  for two ``hidden'' monic polynomials $f,g\in \F_q[X]$:
\begin{quote}
given 
oracles $\Oef$  and   $\Oef$ that on every input $x \in \F_q$ output
$\Oef(x) = f(x)^e$ and $\Oeg(x) = g(x)^e$ for some large positive integer $e\mid q-1$,  
decide whether $f = g$.
\end{quote}

We also consider the following problem {\it Interpolation from Powers\/} for a ``hidden'' monic polynomials $f\in \F_q[X]$:
\begin{quote}
given oracle $\Oef$ that on every input $x \in \F_q$ outputs
$\Oef(x) = f(x)^e$ for some large positive integer $e\mid q-1$,
recover $f$.
\end{quote}

In particular, for a linear polynomial 
$f(X) = X+s$, with a `hidden' $s\in \F_q$, 
we denote  $\Oes=\Oef $.  We remark that in this case 
there are two naive algorithms that work for linear polynomials:
\begin{itemize}

\item One can query $\Oes$ at $e+1$ arbitrary points and then 
using a  fast interpolation
algorithm, see~\cite{vzGG}, obtain a deterministic algorithm of
complexity $e (\log q)^{O(1)}$
(as in~\cite{vzGG}, we measure the complexity of
an algorithm by the number of bit operations in the standard RAM model of computation).

\item For probabilistic testing one can  query $\Oes$ (and $\Oet$) at randomly chosen elements $x \in \F_q$ until the desired level of confidence is achieved (note that the equation $(x+s)^e = (x+t)^e$ has at most $e$ solutions $x \in \F_q$). 

\end{itemize}

These naive algorithms have been improved by Bourgain, Garaev, Konyagin and Shparlinski~\cite{BGKS} in several cases (with respect to both the time complexity and the number of queries). 

Furthermore,  
for linear polynomials  $f(X) = X+s$, Dam,  Hallgren and  Ip~\cite{vDamHallgIp}  provide a quantum 
polynomial time algorithm to find $s$,  see also~\cite{vDam}, in the case of the oracle oracles  $\Oef$, with $e = (q-1)/2$ (and odd $q$).  We remark that querring  this oracle
is equivalent to asking for the quadratic character of the computed value. Thus the oracle returns 
only one bit of information, making this the hardest case.

Russell and 
Shparlinski~\cite{RusShp} have initiated the study of this question
for non-linear monic polynomials and,  in the case of the oracle  $\fO_{(p-1)/2,f}$, 
for a prime $q=p$,
have designed several  classical  and quantum algorithms. 
More recently, several other algorithms, for an arbitrary $e$,  have
been given in~\cite{IKSSS}.  The algorithms from~\cite{IKSSS} usually 
improve on the above trivial  interpolation 
and random sampling algorithms. However in the settings of~\cite{IKSSS}
the degree of the polynomials is assumed to be small.

Here we  concentrate on the case of polynomials large degree
and fields of large characteristic $p$, in particular, on the case of prime 
fields $\F_p$, and use 
a different approach to obtain new results in this case. 

We also observe that  if 
$$
e \max\{\deg f, \deg g\}> q
$$ 
then the above naive  interpolation and random sampling algorithms both fail. 
Indeed, since queries from an extension field are not permitted,
and $\F_q$ may not have enough elements to make these algorithms  work.
This indicates 
that the difficulty of the problem grows with  the degrees of polynomials 
involved.

Our approach is based on a new upper bound on the size of an 
intersection of value set on consecutive integers of a rational function of large degree with a small subgroup 
of a finite field. Results of this type, complementing this of~\cite{GomShp,IKSSS,Shp2}
are of independent interest, see 
also ~\cite{BGKS,Chang,CCGHSZ,CGOS,CSZ, Ost, Shp1} for further 
results on related problems.

\section{Main results}
\label{sec:main res}

First we consider the identity testing case of two unknown {\em monic} 
polynomials $f,g\in \F_q[X]$ of degree $d$ given 
the oracles $\Oef$ and $\Oeg$. We remark that if $f/g$ is
an $(q-1)/e$-th power of a 
 rational function 
over $\F_q$ then  it is impossible to distinguish 
between $f$ and $g$ from the oracles $\Oef$ and $\Oeg$.

We however impose a slightly stronger condition that $f/g$ is not a perfect power
 of a  rational function.

\begin{definition}
\label{def:Large_e} 
We see say that a rational function $\psi \in \F_p(X)$ is a nontrivial perfect power 
if $\psi = \varphi^k$ for some rational function $\varphi \in \F_p(X)$ of positive degree 
and some positive integer $k$. 
\end{definition}

\begin{theorem}
\label{thm:Large_e}  There are absolute constants $c_1, c_2> 0$ such that
for a prime $p$ and a positive integer
$e\mid p-1$, given two oracles $\Oef$ and $\Oeg$  for some 
unknown monic polynomials $f,g\in \F_p[X]$ of degree $d$ such 
that 
\begin{equation}
\label{eq:ed cond}
e \le c_1 \min\left\{ p d^{-3/2} ,   p^{3/2}d^{-7/2}\right\}
\end{equation}
such that $f/g$ is a nontrivial perfect power, 
there is a deterministic algorithm to decide 
whether $f = g$ in at most 
$$c_2 \max\left\{d^{3} e^2 p^{-1} ,   d^{7/3}e^{2/3}\right\}
$$
 queries 
to  the oracles $\Oef$ and $\Oeg$. 
\end{theorem}

Next, we consider the interpolation problem for square-free polynomials.
\begin{theorem}\label{thm:interpol}
For any fixed $\varepsilon>0$ there are constants $c_1,c_2,c_3>0$ such that for a prime $p$ and a positive integer $e\mid p-1$, $e\geq c_1$,  given an oracle $\Oef$ for some unknown monic square-free polynomial $f\in \F_p[X]$ of degree $d\geq 2$ such that
$$
e \le c_2 \min\left\{ p d^{-3/2} ,   p^{3/2}d^{-7/2},p^{1-\varepsilon}\right\},
$$
there is a deterministic algorithm which makes at most 
$$c_3\max\{d^3e^2p^{-1},d^{7/3}e^{2/3}\}
$$ 
queries to  the oracle $\Oef$ and recovers the polynomial $f$ in time 
$$
e^{(1+\varepsilon)d}\max\{e^{3/2}p^{-1+\varepsilon},  e^{1/6}p^{\varepsilon}\}.
$$
\end{theorem}

\section{Points with coordinates from subgroups on plane curves over $\F_p$}

We recall that the
notations $U = O(V)$, $U \ll V$ and  $V \gg U$ are all equivalent to the
statement that the inequality $|U| \le c V$ holds with some
constant $c> 0$. Throughout the paper, any implied constants in these symbols are absolute, in particular, all estimates are uniform with respect to 
the degree $d$, the exponent $e$ and the field characteristic $p$. 

Our argument relies on a result of Corvaja and Zannier~\cite[Corollary~2]{CorZan}.

We also use $\Fbar_p$ for the algebraic closure of $\F_p$

We also write $\deg_X F(X,Y)$ and $\deg_Y F(X,Y)$ for the degree of $F(X,Y)\in \Fbar_p[X,Y]$ in $X$ and $Y$ 
respectively, and reserve $\deg F(X,Y)$ for the total degree. 

We say that $F(X,Y)\in \Fbar_p[X,Y]$ is a {\it torsion polynomial\/}
it is of the form  
\begin{equation}
\label{eq:TorsPoly}
\alpha U^mV^n + \beta \qquad  \text{or}\qquad \alpha U^m + \beta V^n. 
\end{equation}

We remark that as we work over the  algebraically closed field $\Fbar_p$, 
the notions of irreducibility and absolute irreducibility coincide. 

\begin{lemma}
\label{lem:CZ-Irred}
Assume that   $F(X,Y)\in \Fbar_p[X,Y]$ is an
irreducible polynomial with $\deg F = d$
which  is not a 
torsion polynomial. For any multiplicative 
subgroups $\cG, \cH \subseteq  \Fbar_p^*$,  
we have
$$
 \# \left\{(u,v) \in \cG\times \cH:~ F(u,v) = 0\right\} \ll \max\left\{  \frac{d^2W}{p}, \,    d^{4/3}W^{1/3}\right\} , 
$$
where
$$
W = \#\cG \#\cH.
$$
\end{lemma}
\begin{proof}To apply~\cite[Corollary~2]{CorZan} we need to estimate the Euler characteristic $\chi$
of the curve $F(X,Y)=0$ in terms of $d$. For the genus $g$ we have the well-known 
estimate $g \ll d^2$. The set $\cS$ from~\cite[Theorem~2]{CorZan}  corresponding to the 
scenario of~\cite[Corollary~2]{CorZan} is the set of poles and zeros of coordinate
functions $X$ and $Y$ and thuse $\# \cS\ll d$. The result now follows. 
\end{proof}

We now need a version of Lemma~\ref{lem:CZ-Irred} which applies to any curve.

\begin{lemma}
\label{lem:CZ-Arb}
Assume that   $F(X,Y)\in \Fbar_p[X,Y]$ is of degree $\deg F=d$ and is not divisible by a torsion polynomial. For any mutltiplicative 
subgroups $\cG, \cH \subseteq  \Fbar_p^*$, we have
$$
\# \left\{(u,v) \in \cG\times \cH:~ F(u,v) = 0\right\}
\ll \max\left\{ \frac{d^2 W}{p}, \,   d^{4/3}W^{1/3}\right\} , 
$$ 
where
$$
W = \#\cG \#\cH.
$$
\end{lemma}
 
\begin{proof} If $F$ is an irreducible polynomial then  the bound is immediate frow Lemma~\ref{lem:CZ-Irred}.
Otherwise we factor $F$ into irreducible (over $\Fbar_p$ components of degrees, say, $d_1, \ldots, d_s$ and then we obtain 
$$
\# \left\{(u,v) \in \cG\times \cH:~ F(u,v) = 0\right\}
\ll \max\left\{  \frac{ A W}{p}, \,  B W^{1/3}\right\} , 
$$ 
where, by the convexity argument, 
$$
A = \sum_{i=1}^s d_i^2\le  d^2\mand B =   \sum_{i=1}^s d_i^{4/3}\le   d^{4/3}, 
$$ 
which concludes the proof. 
\end{proof}

\section{Non-vanishing of some resultants}

We recall the following well-known statement, see for example~\cite[Lemma~6.54]{LN} (the 
proof extends from polynomials to rational functions without any changes). 

\begin{lemma}
\label{lem:AbsIrred}
Let $f(X), g(X) \in \Fbar_p[X]$ be such that $f/g$ is not a perfect power of a  rational function.
Then for any integer $m \ge 1$ the polynomial $f(X) - Y^mg(X)$ 
is irreducible.
\end{lemma}

\begin{lemma}
\label{lem:Syst Eq}
Let $f(X), g(X) \in \Fbar_p[X]$ be polynomials of degrees at most $d$ and 
be such that $f/g$ is not a perfect power of a  rational function. 
Then for any integers $m,n$ the   system of equations 
$$f(X) - Y^mg(X) =  f(X+a) -  bY^ng(X+a)=0,
$$ with $a,b \in \Fbar_p^*$ 
defines a zero dimensional variety, unless  
$m = \pm n$,  and pairs $(a,b)$ are from a set of cardinality at most $4$.
\end{lemma}

\begin{proof} If $mn=0$ the result is trivial

Changing the roles of $f$ and $g$, we can always assume that $m > 0$. 

Now, let $n > 0$. Since by Lemma~\ref{lem:AbsIrred} both polynomials are 
 irreducible, they may have a common factor if and only if 
they are equal up to a  factor from  $\Fbar_p^*$ and thus $m = n$. 
Furthermore comparing the coefficient at the front of $Y$ we conclude that 
$$
b\(f(X) - Y^mg(X)\) = f(X+a) -  bY^ng(X+a).
$$
Hence $m = n$. Now, comparing the parts which do not depend on $Y$ we obtain $bf(X)  = f(X+a)$ hence $b = 1$ and then $a = 0$. 

We now consider the case $n < 0$ and rewrite the equations as
$$f(X) - Y^mg(X) =  Y^{-n} f(X+a) -  bg(X+a)=0. 
$$
Again, by Lemma~\ref{lem:AbsIrred}, the polynomials involved are  irreducible again, 
hence  $m = -n$.  Comparing the parts which do not depend on $Y$ we see that 
$b$ is uniquely defined and then $a$ is uniquely defined as well. 
\end{proof}

\begin{lemma}
\label{lem:Res_uv} 
Let $f(X), g(X) \in \Fbar_p[X]$ be  polynomials of degrees at most $d$ and 
be such that $f/g$ is not a perfect power of a  rational function.
Then there is a set $\cE \subseteq  \Fbar_p^*$ of cardinality  $\# \cE = O(d^2)$ 
such that for $a\in \Fbar_p^*\setminus \cE $ the resultant 
$$
R_a(U,V) = \Res_X\((f(X) - Ug(X), f(X+a) - Vg(X+a)\)
$$
with respect to $X$, 
is not divisible by a torsion polynomial.
\end{lemma}

\begin{proof} Assuming that 
$$
R_a(U,V) = \Res_X\((f(X) - Ug(X), f(X+a) - Vg(X+a)\)
$$
is divisible by a polynomial 
of the form~\eqref{eq:TorsPoly}, 
we easily derive that there is a variety of the type considered in 
Lemma~\ref{lem:Syst Eq}
which is of positive dimension. 
Since $m \le \deg_U R_a = O(d)$ and $n \le \deg_V R_a = O(d)$, 
the result follows.  
\end{proof}

\section{Intersection of polynomial images of intervals and subgroups}

For a rational function $\psi(X) = f(X)/g(X)\in\Fbar_p(X)$ with two relatively 
primes polynomials $f,g \in \Fbar_p[X]$
and a  set $\cS \subseteq \F_p$, we use $\psi(\cS)$ to 
denote   the value set 
$$
\psi(\cS) = \{\psi(x)~:~x \in \cS, \ g(x) \ne 0\} \subseteq \F_p.
$$
Given  an interval $\cI = [1,H]$ with a positive $H<p$ and a
subgroup $\cG \in \Fbar_p$ 
we consider the size of the intersection of $\psi(\cI)$ and $\cG$, 
that is, 
$$
N_\psi(\cI,\cG) =\#\(\psi(\cI) \cap \cG\).
$$
Bounds on this quantity for various functions $\psi$ is in the background 
of the algorithms of~\cite{BGKS, IKSSS}. These bounds  are also of independent interest as they are natural analogues of the problem
of bounding  
$$
N_\psi(\cI,\cJ) =\#\(\psi(\cI) \cap \cJ\), 
$$
for two intervals $\cI$ and $\cJ$ and similar sets,  which has recently been actively investigated, 
see~\cite{Chang,CCGHSZ,CGOS,CSZ, GomShp, Ost, Shp1,Shp2} and references therein.

\begin{lemma}
\label{lem:fI G}
Let  $\psi(X) = f(X)/g(X)\in\Fbar_p(X)$ with two relatively 
primes   polynomials $f,g \in \Fbar_p[X]$ of degrees at most $d$ 
and  such that $\psi$ is not a perfect power of a  rational function.
Then for  any interval $\cI = [1,H]$ of length  $H<p$ and any
subgroup $\cG \in \Fbar_p$ of order $e$, we have 
$$
N_\psi(\cI,\cG)   \ll H^{1/2}  \max\left\{d^{3/2} e p^{-1/2} ,   d^{7/6}e^{1/3}\right\} . 
$$
\end{lemma}

\begin{proof} 
Denote $M = N_\psi(\cI,\cG)$. Let $\overline \cI = [-H, H]$. Clearly the system of equations (over $\Fbar_p$):
$$
f(x) = u \mand  f(x+y) = v, \quad x\in \cI, \ y  \in \overline \cI, \ u,v \in \cG, 
$$
has at least $M^2$ solutions. 
Let $M_y$ be the number of solutions with a fixed $y$. 
The 
$$
\sum_{y  \in \overline \cI} M_y  \ge M^2.
$$
We choose 
$$
L = \frac{M^2}{2(2H+1)}
$$
 and consider the set $\cY$ of $y  \in \overline \cI$ with 
$M_y> L$. 
Using that $M_y \le H$ we write 
$$
  H \# \cY \ge   \sum_{\substack{y  \in \overline \cI\\ M_y  > L}} M_y 
\ge M^2 -   \sum_{\substack{y  \in \overline \cI\\ M_y \le  L}} M_y 
\ge M^2 - (2H+1) L \ge \frac{1}{2} M^2.
$$
Now if  $\# \cY \le\# \cE$, where $\cE$ is as in Lemma~\ref{lem:Res_uv}. 
then $M^2 \le 2H \# \cE  \ll d^2 H$ and thus $M \ll d H^{1/2}$, which 
is stronger than the desired bound.

Hence we can assume that  $\# \cY > \# \cE$ and thus  there is $ \cY \setminus \cE \ne \emptyset$.
We now fix any  $a \in \cY \setminus \cE$ and consider the systeof equations
$$
f(x) = u \mand  f(x+a) = v, \quad x\in \cI,   \ u,v \in \cG,
$$
Using the resultant to eliminate $x$ we obtain  $R_a(u,v)= 0$ for each solution, 
where $R_a(U,V)$ is as in  Lemma~\ref{lem:Res_uv}. Due to choice of $a$, we see that 
the bound of Lemma~\ref{lem:CZ-Arb} applies, and since for every fixed $u$ there are at most $d$ 
values of $x$ we obtain
\begin{align*}
L  & < M_a \le d \# \left\{(u,v) \in \cG\times \cG:~ R_a(u,v)= 0\right\}\\
& \ll d \max\left\{  \frac{d^2 e^2}{p}, \,    d^{4/3}e^{2/3}\right\}  .
\end{align*}
Recalling the definition of $L$ we obtain 
$$
M^2 \ll H  \max\left\{d^3 e^2 p^{-1} ,  d^{7/3}e^{2/3}\right\} 
$$
and the result follows.  
\end{proof}

\section{Finding solutions to binomial equations}

Consider the equation $x^e=A$ in $\F_p$ for some $A\in\F_p^*$. There exists a polynomial time (polynomial in $\log p$) probabilistic algorithm which finds all solutions of this equation, see, for example,\cite[Theorem~7.3.1]{BachShallit}.
 
If the equation $x^e=A$ satisfies some additional restriction one can derandomize \cite[Theorem~7.3.1]{BachShallit}. Namely, let $\ind x$ denote the index (discrete logarithm) of $x$ with respect to a fixed primitive root $g$ modulo $p$, that is the unique integer $z\in [1,p-1]$ with $x=g^z$. Combining \cite[Lemma~11]{BGKS} and \cite[Corollary~40]{BGKS} we derive:

\begin{lemma}\label{lemma:powers}
Let $\varepsilon>0$, then there exist constants $c_1,c_1>0$ with the following properties. For a prime $p$ and $e \mid p-1$ with $e\geq c_1$, there exist an integer $n\mid p-1$ with $n\leq c_2$ and a deterministic algorithm $\cA_n$ such that  for a given $A\in\F_p$ it finds all solutions of the equation $x^e=A$ satisfying $n \mid \ind x$ in time $e(\log p)^{O(1)}$. One can find $n$ and $\cA_n$ in time $e^{1/2+o(1)}+p^{\varepsilon}(\log p)^{O(1)}$.
\end{lemma}

\section{Number of interpolating polynomials}

In this section we estimate the number of polynomials $f\in\F_p[X]$ such that $f(x_i)^e=A_i$ ($i=0,\ldots, k$) for some pairwise distinct $x_0,\ldots, x_k\in\F_p$ and arbitrary $A_0,\ldots, A_k\in\F_p$.

The following result is a special case of \cite[Lemma~4.1]{VS12}.

\begin{lemma}\label{lemma:shifted_subgroups}
Assume that for a fixed integer $m\geq 1$ we have
\begin{equation}\label{eq:e}
p\geq \(2m \left\lfloor e^{1/(2m+1)} \right\rfloor+2m+2 \)e.
\end{equation}
Then for pairwise distinct $\xi_1,\ldots, \xi_m\in\F_p^*$ and arbitrary $\mu_1,\ldots, \mu_m\in\F_p^*$ the bound
$$
\#\(\cG_e\cap(\mu_1\cG_e+\xi_1)\cap \ldots \cap (\mu_m\cG_e+\xi_m) \)\ll e^{(m+1)/(2m+1)}
$$
holds, where the implied constant depends on $m$.
\end{lemma}

\begin{lemma}\label{lemma:bound_for_poly}
Let $d\geq 2$.
For any fixed $\varepsilon>0$ there exist an integer $m$ and a constant $c>0$ such that for a prime $p$ and a positive integer $e\mid p-1$ with $c<e<p^{1-\varepsilon}$,  for any pairwise distinct $x_0,\ldots, x_{d(m-1)+1}\in\F_p$ and arbitrary $A_0,\ldots, A_{d(m-1)+1}\in\F_p^*$ the number of monic polynomials $f\in\F_p[X]$ of degree at most $d$ such that 
$$
f(x_i)^e=A_i, \quad i=0,\ldots, d(m-1)+1
$$
is at most $e^{d-1/2+\varepsilon}$. 
\end{lemma}

\begin{proof}
Choose $m$ such that the it satisfies~\eqref{eq:e} with $\varepsilon<1/(4m+1)$.  

For $0\leq i\leq d$ let $L_i(X)$, be the $i$th Lagrange interpolation polynomial on the points $x_0,\ldots, x_{d}$, that is
\begin{equation}\label{eq:Lag}
L_i(X)=\prod_{j\neq i}\frac{X-x_j}{x_i-x_j}, \quad 0\leq i\leq d.
\end{equation}
Define $c_i\in\F_p$,  $0\leq i\leq d(m-1)+1$, as $c_i^e=A_i$. Then $f(X)$ has the form
$$
f(X)=\sum_{i=0}^{d} c_i\lambda_i L_i(X)
$$
for some $\lambda_0,\ldots,\lambda_{d}\in\cG_e$. 

Consider the rational function
$$
\psi(X)=\sum_{i=1}^{d} \frac{c_i\lambda_i L_i(X)}{c_0 L_0(X)}\in\F_p(X).
$$
If it is the zero function, then $A_1=f(x_1)^e=\(c_0 \lambda_0 L_0(x_1)\)^e=0$ by~\eqref{eq:Lag}, a contradiction. As both the numerator and the denominator have degree at most $d$, it takes each value at most $d$ times. 
Thus, there are points $x_{j_1},\ldots, x_{j_m}$, $d<j_1<\ldots<j_{m}\leq d(m-1) +1 $ such that $\psi$ takes different values in these points.

For $1\leq \ell\leq m$ we have
$$
f(x_{j_\ell})=\sum_{i=0}^{d} c_i\lambda_i L_i(x_{j_\ell})=c_{j_\ell}\lambda_{j_\ell}
$$
for some $\lambda_{j_1},\ldots,\lambda_{j_m}\in\cG_e$. Then
$$
 \lambda_0=\lambda_{j_\ell}\frac{c_{j_\ell}}{c_0 L_0(x_{j_\ell})}-   \psi(x_{j_\ell}),   \qquad  1\leq \ell \leq m, 
$$
so
$$q
\lambda_0 \in \bigcap_{\ell=1}^{m}\(\mu_{\ell}\cG_e-\xi_\ell\)
$$
with
$$
\mu_{\ell}= \frac{c_{j_\ell}}{c_0 L_0(x_{j_\ell})}\in\F_p^* \quad \text{and} \quad \xi_\ell=\psi(x_{j_\ell})  \in\F_p^*, \quad 1\leq \ell \leq m.
$$

For fixed $\lambda_1,\ldots, \lambda_d\in\cG_e$, there are at most $e^{1/2+\varepsilon}$ choices for $\lambda_0$ by Lemma~\ref{lemma:shifted_subgroups}. As for each monic polynomials there are exactly $e$ non-monic polynomials with the same property, we obtain the result.
\end{proof}

\begin{lemma}\label{lemma:frac_interpol} 
Let $d\geq 1$. For any pairwise distinct $x_1,\ldots, x_{2d}\in\F_p$, arbitrary $y_1,\ldots, y_{2d}\in\F_p^{*}$ and $h\in\F_p^*$, there is at most one monic polynomial $f\in\F_p[X]$ of degree at most $d$ such that
\begin{equation}\label{eq:frac_interpol}
\frac{f(x_i)}{f(x_i+h)}=y_i, \quad \text{for } i=1,\ldots, 2d.
\end{equation}
\end{lemma}

\begin{proof}
If $f(X)$ and $g(X)$ are two distinct polynomials with~\eqref{eq:frac_interpol}, then
$$
\frac{f(X)}{f(X+h)}-\frac{g(X)}{g(X+h)}=\frac{f(X)g(X+h)-g(X)f(X+h)}{f(X+h)g(X+h)}
$$
vanishes on $x_1,\ldots, x_{2d}$. As the degree of the nominator is at most $2d-1<\#\cI$, we have $f(X)=g(X)$.
\end{proof}

\begin{lemma}\label{lemma:LES}
 Let $d\geq 1$. For any pairwise distinct $x_1,\ldots, x_{\ell}\in\F_p$, arbitrary $y_1,\ldots, y_{\ell}\in\F_p^{*}$ and $h\in\F_p^*$ consider the matrices
$$
M_s=
\begin{pmatrix}
y_{1}(x_1+h)^{k} -x_1^{k} & \ldots & y_{1}(x_1+h)-x_1  & y_1-1\\
\vdots & \ddots & \vdots & \vdots \\
y_{s}(x_s+h)^{k} -x_s^{k} & \ldots & y_{s}(x_s+h)-x_s  & y_s-1\\
\end{pmatrix}  \in\F_p^{s \times (k+1)}
$$
for $1\leq s\leq \ell$. If $M_{\ell-1}$ has rank $r$, then $M_{\ell}$ has rank $r$ either for at most one or all choices of $y_\ell$.
\end{lemma}

\begin{proof} 
Consider the submatrices $N\in\F_p^{(k+1)\times(k+1)}$ of $M_{\ell}$ which are not submatrices of $M_{\ell-1}$. The determinant $\det N$ is either does not depend on $y_{\ell}$, or a linear polynomial of it. If all such determinants are zero, then $M_{\ell}$ has rank $r$ independently of the choice of $y_\ell$. Otherwise, there is only one value for $y_\ell$ such that the determinant of a submatrix $N$ vanishes. 
\end{proof}

\section{Proof of Theorem~\ref{thm:Large_e}}
\label{sec:proof}

First we note that  there are absolute constants $c_1, c_2> 0$ such that 
if the condition~\eqref{eq:ed cond} is satisfied 
then for 
\begin{equation}
\label{eq:choice H}
H = \fl{c_1 \max\left\{d^{3} e^2 p^{-1} ,   d^{7/3}e^{2/3}\right\} }
\end{equation}
we have $H < p$ and 
then under the conditions of Lemma~\ref{lem:fI G}  we have 
\begin{equation}
\label{eq:N<H}
N_\psi(\cI,\cG)  < H, 
\end{equation}
for any rational function $\psi \in \Fbar_p(X)$ of degree at most $d$, which is a nontrivial perfect power.

Now, since $f/g$ is a nontrivial perfect power
 Lemma~\ref{lem:fI G} applies to $\psi = f/g$. Hence taking $H$ as in~\eqref{eq:choice H}, 
 so the inequality~\eqref{eq:N<H} is satisfied, we see that querring 
 $\Oef$ and $\Oeg$ for $x =1, \ldots, H$,  we have $\Oef(x) \ne \Oeg(x)$ 
 unless $f=g$.

\section{Proof of Theorem~\ref{thm:interpol}} 

\subsubsection*{Step 1.}
Let $n$ and $\cA_n$ be as in Lemma~\ref{lemma:powers}

We call the oracle $\Oef$ on the inputs $x=0,\ldots, (2d-1)n^2+n$ and let $A_x=\Oef(x)$ for $x=0,\ldots, (2d-1)n^2+n$. Then we recover all the polynomials $f\in\F_p[X]$ such that
$$
f(x)^e=A_x, \qquad x=0,\ldots,(2d-1)n^2+n.
$$

If $A_y=0$ for some $y$, then $y$ must be a zero of $f$, thus it is enough to find all $g(X)$ of degree at most $d-1$ such that $g(X)^e=A_x/(x-y)^e$. Thus we can assume, that 
$$
A_x \ne 0, \qquad x=0,\ldots,(2d-1)n^2+n.
$$

For all $i$ with $0\leq i \leq (2d-1)n+1$, by the pigeonhole principle there are $x_i$ and $h_i$ with $in\leq x_i<x_i+h_i\leq (i+1)n$ such that 
$$
\ind f(x_i+h_i)  \equiv \ind f(x_i) \mod n.
$$
Specially, $1\leq h_i\leq n$, so there is a subset $\cI\subset\{i: 0\leq i \leq (2d-1)n+1 \}$  of size $\#\cI=2d$ and $1\leq h\leq n$ such that $h_i=h$ for $i\in\cI$, that is, 
$$
\ind f(x_i+h)  \equiv \ind f(x_i) \mod n,\quad i\in\cI.
$$

To find $\cI$, we just try all pairs $A_x,A_{x+h}$. 

By Lemma~\ref{lemma:powers}, we can extract all such $y_i$ that $y_i^e=A_{x_i}/A_{x_i+h}$ for $i\in\cI$.

\subsubsection*{Step 2.}
In order to obtain candidates for $f$, we consider  the system of 
equations  for the coefficients $f_0,\ldots, f_{d-1}$ of $f$
\begin{equation}\label{eq:LES}
\begin{split}
 x_{i_j}^{d}  + \sum_{k=0}^{d-1} f_k x_{i_j}^{k}  =y_{i_j}&\((x_{i_j}+h)^d +  \sum_{k=0}^{d-1} f_k(x_{i_j}+h)^k\),  \\
 & 1\leq j\leq \ell, 
\end{split}
\end{equation}
for some $1\leq i_1<\dots<i_\ell\leq d$ and fixed $y_{i_1},\ldots, y_{i_\ell}$, where  $1\leq \ell \leq d$.

Now we proceed by induction on $1\leq\ell\leq d$ to obtain a full rank system of equations~\eqref{eq:LES}.

For $\ell=1$ we choose $i_1=1$. As $h\neq 0$, \eqref{eq:LES} is nontrivilal for any choice of $y_{i_1}$.

Next, assume that we have fixed the values $y_{i_1},\ldots, y_{i_\ell}$ for some $1<\ell<d$ such that the~\eqref{eq:LES}
has a full rank. In the following we choose the index $i_{\ell+1}\leq 2d$ and the value $y_{i_{\ell+1}}$ such that the system of equations~\eqref{eq:LES} remains  full rank.

We start with $i_{\ell+1}=i_{\ell}+1$. If the system of equations~\eqref{eq:LES}  is singular independently of the choice of $y_{i_{\ell+1}}$, we increase the index $i_{\ell+1}$. Otherwise, we fix the value of $y_{\ell+1}$. If~\eqref{eq:LES} becomes singular with this choice,  we also increase the index $i_{\ell+1}$.  Then we obtain either a full rank system of equations ~\eqref{eq:LES}, or  $i_{\ell+1}=2d$, but~\eqref{eq:LES} is singular. In the latter case the system of equations does not have solution by Lemma~\ref{lemma:frac_interpol}, so we terminate.

Finally, we terminate if we reach $\ell=d$.

If we have a regular linear system of equations with $\ell=d$, we solve it for $f_0,\ldots, f_{d-1}$ and test whether $f(x)^e=A_x$ for every $x=0,\ldots, (2d-1)n^2+n$. 

When we terminate, we have fixed the values $y_{i_1}, \ldots, y_{i_k}$ with $1\leq  k\leq 2d$. Let $R$ be the set of indices $j$ such that $i_j$ appears in~\eqref{eq:LES}, let $S$ be the set of indices $j$ such that  $y_{i_j}$ is fixed, but the equation belongs to this index leads to a singular system of equations. Finally, let $T$ be the set of indices $j$ such that  $y_{i_j}$ is not fixed. Then $R\cup S\cup T =\{1,\dots, 2d\}$. 

By Lemma~\ref{lemma:LES}, the value $y_{i_j}$, $j\in S$, is uniquely determined by $y_{i_1}, \ldots, y_{i_{j-1}}$ and $i_1, \ldots, i_{j}$. Thus for fixed $R$,  $S$, $T$ there are at most $e^d$ choices for the tuples $(y_{i_1}, \ldots, y_{i_k})$ . We can choose $R,S,T$ in at most $3^{2d}$ ways.
 Thus the running time of this step is $d^2\, e (\log p)^{O(1)}  + d^2 3^{2d}e^d(\log p)^{O(1)}$.

\textit{Step 3.} Let $\cF$ be the set of all possible polynomials obtained in the previous step and let $m$ as in Lemma~\ref{lemma:bound_for_poly}. Let $A_x$ the output of the oracle $\Oef(x)$  for $x=0,1\ldots,d(m-1)+1$ and put
$$
\widetilde \cF=\{g\in\cF~:~g^{e}(x)=A_x \ 0\leq x\leq d(m-1)+1\}.
$$
By Lemma~\ref{lemma:bound_for_poly} we have $\#\widetilde \cF\leq e^{d-1/2+\varepsilon}$.
We can assume, that all of them are square-free. This reduction is done in time $e^d d^2 (\log p)^{O(1)}$. Finally, we use the identity testing algorithm of Theorem~\ref{thm:Large_e} for the oracles $\mathfrak{D}_{e,f}$ and $\mathfrak{D}_{e,g}$ for all $g\in\cF$. 

The running time of the first step is bounded by
$$
e^{1/2+o(1)}+p^{\varepsilon}(\log p)^{O(1)}+d^2\, e (\log p)^{O(1)}
$$
of the second step is bounded by
$$
d^2 3^{2d}e^d(\log p)^{O(1)} 
$$
and of the third step is bounded by
$$
d3^{2d}e^d(\log p)^{O(1)}+
e^{d-1/2+\varepsilon}(\log p)^{O(1)}\max\{d^2e^2p^{-1},d^{7/3}e^{2/3} \}.
$$
Replacing $\varepsilon$ with $\varepsilon/2$ and having $3^{2d}\leq e^{\varepsilon d}$,  we obtain the result.

 
%
\section*{Acknowledgement}

The authors are very grateful to Umberto Zannier for some clarifications 
concerning the results of~\cite{CorZan}. 

The the research of M.K. was supported in part by DFG   
 grants and the Hausdorff grant EXC'59-1, of L.M. 
 was is supported by the Austrian Science Fund (FWF): Project I1751-N26,
 and  
 of I.S. was  supported in part by  
 the Australian Research Council Grant DP170100786. 
%

%

\begin{thebibliography}{MvOV10}

\bibitem{BachShallit} E. Bach and J. Shallit, Algorithmic Number Theory, MIT Press, Cambridge, MA, 1996.



\bibitem{BGKS}
J.~Bourgain, M.~Z.~Garaev, S. V. Konyagin and I. E. Shparlinski,
`On the hidden shifted power problem',
{\it SIAM J. Comp.\/}, {\bf 41} (2012),   1524--1557.
  
 \bibitem{Chang} M.-C. Chang, `Sparsity of the intersection of polynomial
images of an interval', {\it Acta Arith.\/}, {\bf 165} (2014), 243--249.


 \bibitem{CCGHSZ} M.-C. Chang, J. Cilleruelo, M. Z. Garaev, J.  Hern\'andez,
I. E. Shparlinski and A. Zumalac\'{a}rregui,
`Points on curves in small boxes and applications',
{\it Michigan Math. J. \/}, {\bf 63} (2014), 503--534. 
  
 \bibitem{CGOS} J. Cilleruelo, M. Z. Garaev,  A. Ostafe and
I. E. Shparlinski,
`On the concentration of points of polynomial maps
and applications',
{\it Math. Zeit.\/}, {\bf 272} (2012), 825--837.

\bibitem{CSZ} J. Cilleruelo, I. E. Shparlinski and A. Zumalac\'{a}rregui,
`Isomorphism classes of elliptic curves over a finite field in some
thin families',  {\it Math. Res. Letters\/}, {\bf 19} (2012), 335--343.


%
  
  \bibitem{CorZan} P. Corvaja  and U. Zannier, 
 `Greatest common divisors of $u-1$, $v-1$ in positive characteristic 
and rational points on curves over finite fields', 
{\it J. Eur. Math. Soc.\/} {\bf 15} (2013), 1927--1942.


\bibitem{vDam} W. van Dam,
`Quantum algorithms  for weighing matrices and quadratic residues',
{\it Algorithmica\/}, {\bf 34} (2002), 413--428.

\bibitem{vDamHallgIp} W. van Dam, S. Hallgren and L. Ip,
`Quantum algorithms for some hidden shift problems', {\it SIAM J.
Comp.\/}, {\bf 6} (2006),  763--778.

\bibitem{vzGG}
J. von zur Gathen and J. Gerhard, {\it Modern computer algebra\/},
Cambridge University Press, Cambridge, 2013.
  
\bibitem{GomShp} D. G\'omez-P\'erez and I. E. Shparlinski,
`Subgroups generated by rational functions 
in finite fields',  {\it Monat. Math.\/}, {\bf 176} 
(2015), 241--253.

 

\bibitem{IKSSS} G. Ivanyos, M. Karpinski, M. Santha, N. Saxena 
and  I. E.~Shparlinski, `Polynomial interpolation and identity testing 
from high powers over finite fields',   
{\it Algorithmica\/}, {\bf  80} (2018), 560--575.
%
  
\bibitem{LN} R. Lidl and H. Niederreiter,  {\it Finite Fields\/},
Cambridge Univ. Press, Cambridge, 1997.


\bibitem{Ost} A. Ostafe, 
Polynomial values in affine subspaces over finite fields,
\textit{J. D'Analyse Math.\/}, (to appear). 

\bibitem{RusShp} A. C. Russell and I. E. Shparlinski,
`Classical and quantum algorithms
for function reconstruction via character evaluation',
{\it J. Compl.\/},  {\bf 20} (2004), 404--422.

\bibitem{Shp1}
I. E. Shparlinski, `Products with variables from low-dimensional affine
 spaces and shifted power identity testing in finite fields', 
{\it J. Symb.  Comp.\/},  {\bf 64} (2014), 35--41.
  
  \bibitem{Shp2}
I. E. Shparlinski, `Polynomial values in small subgroups of finite fields', 
{\it Revista Matem. Iber.\/},  {\bf 32} (2016), 1127--1136.


\bibitem{VS12}  I. V. Vyugin and  I. D.  Shkredov,
`On additive shifts of multiplicative subgroups', 
{\it Mat. Sb.\/} {\bf 203} (2012), no. 6, 81--100

%


\end{thebibliography}

\end{document}